\documentclass[11pt]{article}

\usepackage{arxiv}

\usepackage[T1]{fontenc}
\usepackage{hyperref}
\usepackage{xcolor}

\usepackage[normalem]{ulem}
\usepackage{amsthm}
\usepackage{amsmath}
\usepackage{amssymb}
\usepackage{mathrsfs}
\usepackage{mathtools}
\usepackage{graphicx}
\usepackage{enumerate}
\usepackage{comment}
\usepackage{algorithm}
\usepackage{algpseudocode}
\usepackage{color}
\usepackage{float}
\usepackage{appendix}
\usepackage{subcaption}

\newtheorem{definition}{Definition}
\newtheorem{remark}{Remark}
\newtheorem{theorem}{Theorem}
\newtheorem{lemma}{Lemma}

\def \be{\begin{equation}}
\def \ee{\end{equation}}

\begin{document}

\title{The Needle is a Thread: \\ Finding Planted Paths in Noisy Process Trees}

\author{
Maya Le\thanks{Dept of Mathematics and Statistics, University of Ottawa, Ottawa, Canada; \texttt{mle038@uOttawa.ca}},
~~Pawe\l{} Pra\l{}at\thanks{Dept of Mathematics, Toronto Metropolitan University, Toronto, Canada; \texttt{pralat@torontomu.ca}},
~~Aaron Smith\thanks{Dept of Mathematics and Statistics, University of Ottawa, Ottawa, Canada; \texttt{asmi28@uOttawa.ca}},
~~Fran\c{c}ois Th\'{e}berge\thanks{Tutte Institute for Mathematics and Computing, Ottawa, Canada; \texttt{theberge@ieee.org}}
}
\date{}


\maketitle    

\begin{abstract}
Motivated by applications in cybersecurity such as finding meaningful sequences of malware-related events buried inside large amounts of computer log data, we introduce the ``planted path'' problem and propose an algorithm to find fuzzy matchings between two trees. This algorithm can be used as a ``building block'' for more complicated workflows. We demonstrate usefulness of a few of such workflows in mining synthetically generated data as well as real-world ACME cybersecurity datasets.
\end{abstract}



\section{Introduction}
\label{sec:intro} 

In many scientific contexts, we observe parts of many large, noisy, labelled directed acyclic graphs and wish to find a small, meaningful path that is common to many of the graphs. In cybersecurity, we might observe the very large ``process trees" of compromised computers and attempt to find the sequence of processes used by an attacker. In supply chain management, we might observe the sequence of sites that a large collection of defective or dangerous products and components went to and attempt to find the source of the problem by finding a common path. In biology, we might observe the genealogical trees of cancer cells or viruses and attempt to find the sequence leading to a drug-resistant or virulent version. In software reverse engineering, we might see many call traces and wish to find the common path associated with a pernicious bug. 

All of these situations have a common mathematical structure: the main goal is to recover a ``planted path" from a large tree. In this paper, we introduce a simple algorithm for extracting paths in Section~\ref{SubsecSimpleAlg} and describe several algorithms for incorporating this algorithm into larger machine-learning workflows in Section~\ref{SecWorkflows}. In Section~\ref{SecSynthToy} we introduce a simple data-generating process and a more formal ``planted labelled path" problem that is loosely analogous to the popular ``planted clique" problem (see~\cite{pclique}) and related ``planted structure" problems (there are many variants; see \textit{e.g.}~\cite{JMLRplant,wu2018statisticalproblemsplantedstructures}). Throughout the paper, we focus on developing realistic models and workflows for messy data over developing statistically-optimal solutions for specific clean versions of the ``planted path" problem. 

While the ``planted path" problem occurs in several areas, we were motivated by the problem of finding meaningful sequences of malware-related events buried inside large amounts of computer log data, and we use this as our main illustrative application. Recall that the basic task in cybersecurity triage is to flag a small number of ``suspicious" lines of a very large log file for careful inspection by an expert (see \textit{e.g.} the survey~\cite{repCyber}). The challenge is that the number of individual log lines that look suspicious under naive heuristics is typically still far too much for human experts to inspect, while the number of ``truly bad" events is quite small. However, it is well-known that \textit{certain families} of cybersecurity events (see \textit{e.g.} the sequence from ``Reconnaissance" to ``Impact" in the MITRE framework\footnote{\url{https://attack.mitre.org/}}) follow our ``planted path" structure quite closely.

We acknowledge that the planted path problem is not a perfect abstraction. In real data, the full event might not be strictly contained in a single path and the noisiness of our observations will typically prevent us from observing a full planted path. We argue that the first of these problems is often minor. In Section~\ref{SubsecACMEExpts} we examine the ACME4 dataset\footnote{\url{https://gdo168.llnl.gov}} and see that the extracted paths capture a meaningful part of the true event, even though they are not the \textit{full} event. In Section \ref{SubsecSmallVar} we describe a small adjustment to our algorithm that allows us to extract non-paths efficiently in the case that this is necessary. The second problem is more important, but even very imperfect filtering can boost the signal of actually-bad sequences (by aggregating events that are meaningfully related) and filter noise (by removing the vast majority of lines that could not belong to any plausible sequence). It has been widely-recognized that some sort of signal-aggregation is critical for creating statistically powerful detectors ~\cite{simplefusion20,optcAnalysis}.

This paper is an exploratory work, showing that path extraction is both feasible and useful for real cybersecurity data. The following version will give further details on more realistic machine learning workflows, comparisons to other algorithms on real datasets, and some basic theory.

\section{Fuzzy Matching Algorithm} \label{SecFuzzyMatch}

We start by introducing the notation that will be used in this section.

\begin{enumerate}
\item Two directed trees $\mathcal{G}$ and $\mathcal{H}$ on sets of nodes $[n] \equiv \{0,1,\ldots,n\}$ and, respectively, $[m]$. Nodes $0$ are the roots of the corresponding trees. We use $anc_{\mathcal{G}}$ to denote the function mapping a non-root node to its unique ancestor, that is, for any $v \in [n] \setminus \{0\}$, $anc_{\mathcal{G}}(v) \in [n]$. Similarly, $desc_{\mathcal{G}}$ denotes the function mapping a node to its set of descendants, that is, for any $v \in [n]$, $desc_{\mathcal{G}}(v) \subseteq [n]$. In particular, $desc_{\mathcal{G}}(v) = \emptyset$ if and only if $v$ is a leaf. Define $anc_{\mathcal{H}}, desc_{\mathcal{H}}$ analogously for $\mathcal{H}$.
\item A set of features $\mathcal{S}$ and two label functions $\phi_{\mathcal{G}} \, : \, [n] \mapsto \mathcal{S}$ and $\phi_{\mathcal{H}}\, : \, [m] \mapsto \mathcal{S}$.
\item A weight function $w \, : \, \mathcal{S}^{2} \mapsto [0,\infty)$. 
\end{enumerate} 

For clarity of exposition, we assume that each node is characterized by a single feature from the set $\mathcal{S}$. While most real-world datasets store many features for each node, this assumption simplifies the notation but there is no loss of generality. (If \textit{e.g.}\ a data structure has many features $\mathcal{S}_{1},\ldots,\mathcal{S}_{k}$, we can simply write $\mathcal{S} = \mathcal{S}_{1} \times \ldots \times\mathcal{S}_{k}$.) Within this framework, the weight function $w$ serves to quantify the similarity between any two features.

\subsection{Definitions} \label{SubsecNotation}

We introduce the notation used to describe partial matches of \textit{single paths} in each tree. We use the partial ordering of nodes in a tree, formally described as follows:

\begin{definition} [Trees and Orderings]
For a directed tree $\mathcal{T}$, define a partial ordering $\leq_{\mathcal{T}}$ on the nodes of $\mathcal{T}$ by writing $u <_{\mathcal{T}} v$ if there is a directed path from $u$ to $v$.
\end{definition}

Note that there is a unique minimal element, the root of $\mathcal{T}$, but there could be many maximal elements corresponding to leaves of $\mathcal{T}$. We also note that it is always possible to order the set of nodes $[n]$ of $\mathcal{T}$ so that $u <_{\mathcal{T}} v$ implies $u < v$. There are usually many such orderings that can be easily found by, for example, running the Breadth-First Search (BFS) or the Depth-First Search (DFS) algorithms and using the discovery time to order nodes. Let us note that the reverse implication does not hold unless $\mathcal{T}$ is in fact a single path. 

For the remainder of the paper, we assume that the nodes of considered trees $\mathcal{G}$ and $\mathcal{H}$ are ordered in the above way. Our goal is to find an optimal matching of nodes in our two trees (with respect to the fixed weight function $w$). Formally, we introduce the following definition:

\begin{definition} [Valid (Partial) Matchings] \label{DefValidPartialMatchingPaths}
We say that a sequence of pairs of nodes $(u_{1},v_{1}), \ldots, (u_{k},v_{k}) \in [n] \times [m]$ is a valid matching between trees $\mathcal{G}$ and $\mathcal{H}$ if $u_{i} <_{\mathcal{G}} u_{i+1}$ and $v_{i} <_{\mathcal{H}} v_{i+1}$ for all $i \in \{1,2,\ldots,k-1\}$.

We say that such a sequence is a valid \textit{partial} matching up to $(u,v) \in [n] \times [m]$ if it is a valid matching \textit{and furthermore} $u_{k} \leq_{\mathcal{G}} u$ and $v_{k} \leq_{\mathcal{H}} v$.
Let us denote by $\mathcal{P}_{u,v}$ the set of valid partial matchings up to $(u,v)$ and by $\mathcal{P} = \bigcup_{u,v} \mathcal{P}_{u,v}$ the set of all valid matchings.
\end{definition}

With the above definition at hand, we can then define the \textit{score} of a valid matching $p =(u_{1},v_{1}), \ldots, (u_{k},v_{k})$ as follows:
\be \label{SequenceScore}
s(p) = \sum_{i=1}^{k} w(\phi_{\mathcal{G}}(u_{i}),\phi_{\mathcal{H}}(v_{i})).
\ee 
Finally, the \emph{similarity score} between trees $\mathcal{G}$ and $\mathcal{H}$ is simply the largest possible score over all valid matchings: 
\be \label{TreesScore}
s( \mathcal{G}, \mathcal{H}) = \max_{p \in \mathcal{P}} s(p).
\ee 

\subsection{Matching Algorithm}\label{SubsecSimpleAlg}

In this subsection, we give a ``fuzzy matching'' algorithm for computing $s( \mathcal{G}, \mathcal{H})$, the feature-based similarity score between $\mathcal{G}$ and $\mathcal{H}$ introduced above---see Algorithm~\ref{alg:basic_match}. The algorithm is very similar to the dynamic programming algorithm of \cite{SM_81}. In particular, we will use the bottom-up (tabulation) variant of the problem which avoids recursion. The time complexity of this algorithm is clearly polynomial: $O(nm)$. 

Throughout this subsection, we use the convention that the root node $0$ has a unique ancestor $-1$, which is never matched. This convention is used purely to avoid introducing special corner cases to the algorithm.

\begin{algorithm}[ht]
\footnotesize
\caption{Basic Matching Algorithm}\label{alg:basic_match}
\begin{algorithmic}[1]
\algrenewcommand\alglinenumber[1]{} 
\State \textbf{Input:} Trees $\mathcal{G}, \mathcal{H}$ on sets of nodes $[n]$ and $[m]$, label functions $\phi_{\mathcal{G}} \, : \, [n] \mapsto \mathcal{S}$, $\phi_{\mathcal{H}} \, : \, [m] \mapsto \mathcal{S}$, weight function $w \, : \, \mathcal{S}^{2} \mapsto [0,\infty).$

\algrenewcommand\alglinenumber[1]{\scriptsize #1}
\setcounter{ALG@line}{0}
\State Initialize $A_{-1,-1} = A_{-1,v} = A_{u,-1} = 0$ for all $u \in [n]$, $v \in [m]$.
\For{$u \in [n]$:}
\For{$v \in [m]$:}
\State Set 
\[
A_{u,v} = \max \left\{
    \begin{array}{lll}
        A_{anc_{\mathcal{G}}(u), v}  \\
        A_{u,anc_{\mathcal{H}}(v)}  \\
        w(\phi_{\mathcal{G}}(u),\phi_{\mathcal{H}}(v)) + A_{anc_{\mathcal{G}}(u), anc_{\mathcal{H}}(v)}
    \end{array}
\right\}
\]
and $C_{u,v} \in \{1,2,3\}$ according to which of these three options was selected (breaking ties in favour of the largest option). 
\EndFor
\EndFor

\State Initialize $\gamma_{\mathcal{G}},\gamma_{\mathcal{H}}$ to be empty lists and $\ell = \textrm{argmax}_{(u,v)} A_{u,v}$.
\While{$\ell_{1} \neq -1$ and $\ell_{2} \neq -1$}: 
\If{$C_{\ell} = 3$}: 
\State Prepend $\ell_{1}$ to  $\gamma_{\mathcal{G}}$ and $\ell_{2}$ to  $\gamma_{\mathcal{H}}$.
\EndIf 
\State Set
\[
\ell  = 
    \begin{cases}
        (anc_{\mathcal{G}}(\ell_{1}),\ell_{2}), \qquad & C_{\ell} = 1 \\
        (\ell_{1},anc_{\mathcal{H}}(\ell_{2})), \qquad & C_{\ell} = 2  \\
         (anc_{\mathcal{G}}(\ell_{1}),anc_{\mathcal{H}}(\ell_{2})), \qquad & C_{\ell} = 3.
    \end{cases}
\]
\EndWhile
\State Let $L$ be the length of $\gamma_{\mathcal{G}}$. Return the sequence $a_{1:L} = (\gamma_{\mathcal{G}}(1),\gamma_{\mathcal{H}}(1)),\ldots,(\gamma_{\mathcal{G}}(L),\gamma_{\mathcal{H}}(L))$ and the score $A = \max_{u,v}(A_{u,v})$.
\end{algorithmic}
\end{algorithm}

Figure~\ref{fig:simple_matching} illustrates a matching of two labelled trees (sampled according to the procedure in Algorithm~\ref{alg:stm} from the next section) that were matched with Algorithm~\ref{alg:basic_match}. The sequences in the ``planted" paths are independent noisy samples from the ground-truth sequence 4-4-2-4-0-2-2-2-3-1 and are placed on the right-hand sides of the trees. The optimal match in the observed data is the sequence 4-0-2-4-2-2-2-4-3-1, and the matched nodes are exactly along the true planted path. Unmatched nodes are coloured grey; matched nodes in the two trees are coloured with same colour. Note that the ``alphabet" of symbols at each node is very small and the planted sequence has noise, so there are both some coincidental ``false positive" matches and some accidental ``false negative" matches along the ground-truth planted path. Nonetheless, the highest-scoring match is a subset of the path with the planted sequence, which is on the right-hand side of the trees as displayed.

\begin{figure}[ht]
    \centering
    \includegraphics[width=0.8\linewidth]{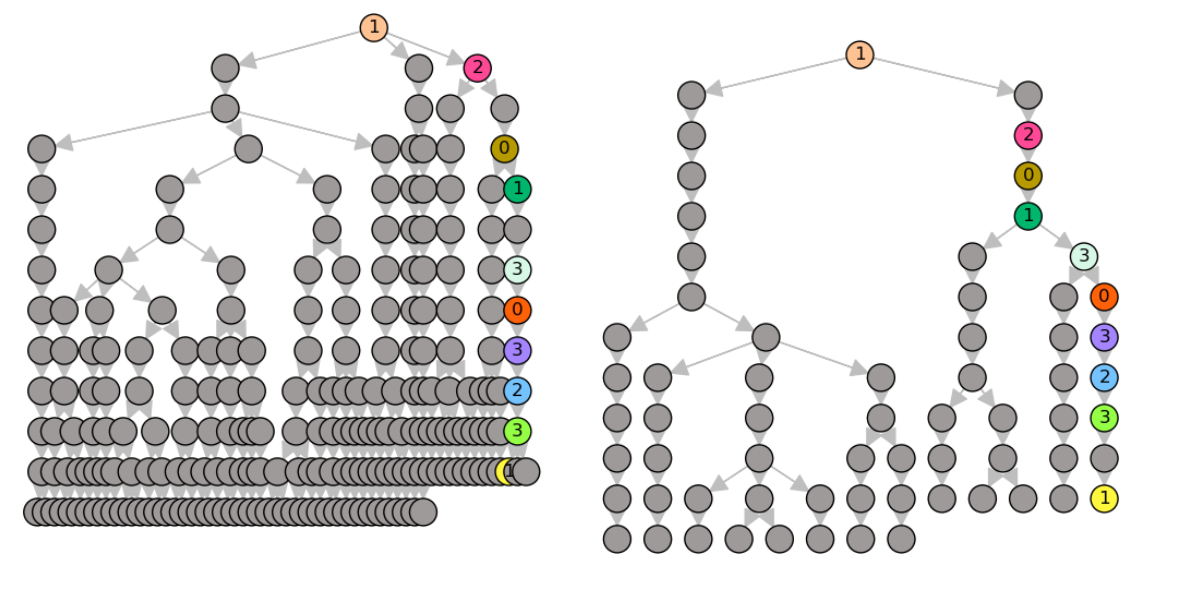}
    \caption{Example of a matching of two trees.} \label{fig:simple_matching}
\end{figure}

\begin{remark} [When do we expect to find the ``ground truth" path?]
It might be surprising that we perfectly recover the planted path in Figure~\ref{fig:simple_matching}, even though the actual matched sequence is not very close to the ground-truth planted sequence.  To aid intuition, we give a rough heuristic argument for when we expect the highest-scoring match to be (nearly) a subset of the planted path. 

Assume that we have a depth-$d$ $m$-ary tree in which a given length-$d$ sequence is planted; other node labels are chosen i.i.d.\ uniformly from an alphabet of size $k$. Also assume the following simple noise model: labels in the planted path are re-sampled uniformly from the alphabet with some ``error" probability $0 < p < 1$.

Let us now consider the optimal match between two randomly labelled trees. Each node along the planted path matches with probability $q \equiv (1-p)^{2} + (1-(1-p)^{2}) k^{-1}$, while nodes \textit{not} on the planted path match with probability $k^{-1}$. Thus, the match of the planted path has score roughly $qd \pm O(\sqrt{qd})$, while nearly-disjoint paths have scores roughly $k^{-1}d \pm O(\sqrt{k^{-1}d})$. There are roughly $m^{d}$ non-planted paths and they are close to disjoint, so standard large-deviation results for Gaussians suggest that the \textit{highest} non-matching path has score roughly $k^{-1}d + \sqrt{k^{-1}d} \sqrt{2 \log(m )d }$. Thus, we expect recovery roughly when
\be 
\left(  (1-p)^{2} + (1-(1-p)^{2}) k^{-1}\right) d \gg \left(k^{-1} + \sqrt{2 \log(m ) k^{-1} } \right) d.
\ee 
In other words, we expect the planted path to have the highest score even when the branching factor $m$ is fairly large and the alphabet size $k$ is fairly small.
\end{remark}

When Algorithm~\ref{alg:basic_match} is applied with score $\mathcal{S} \equiv 1$, it returns the largest common sub-path of the two trees. It can be viewed as a reasonable ``default" behaviour for a planted-path algorithm, but in practice one may want slightly different behaviour including 
(i) matching subtrees instead of paths, 
(ii) preference for shorter matches,
(iii) reduce the number of ``skipped'' nodes, 
(iv) the ``top-$K$" paths, rather than the single best match. 
Moreover, some matches may be more informative than others. A standard default choice is to ``fit" the score function $w$ to the data, setting
    \be \label{SimpleProbScoring}
    w(i,i) \propto \frac{2}{N(i)+2},
    \ee where $N(i)$ is the number of times label $i$ appears in the dataset (see \textit{e.g.} Chapter~8 of \cite{Manning_Raghavan_Schutze_2008} for an introduction to scoring rules of this nature). 

\section{Workflows} \label{SecWorkflows}

We expect Algorithm~\ref{alg:basic_match} to be a useful ``building block'' used as a step inside more complicated workflows. To show its power, we introduce below two that seem particularly relevant. More workflows will be discussed in the journal version of this paper. 

\subsection{Matching Workflow: Unsupervised Setting} \label{SecWorkUns}

One use of Algorithm~\ref{alg:basic_match} is to find a known, small template $\mathcal{H}$ inside of a large observed graph~$\mathcal{G}$. In applications, we may not know a ``good" template even if we have strong reason to believe a good template exists. In this situation, we may consider the problem of finding good matchings between two large graphs, then extracting candidate templates. The main algorithm here is a simple embedding and clustering algorithm based on the matching score. 

As is typically the case for workflows that involve embedding and clustering, there is quite a bit of freedom to ``switch out" components of our algorithm. As such, we give a somewhat-generic version of the workflow, Algorithm~\ref{alg:unsupervised_clustering}, with the following components being user-specified:

\begin{enumerate}
    \item \textbf{Normalizer:} This is a function $\mathcal{N}$ from an $n$ by $n$ ``similarity" matrix to an $n$ by $n$ ``distance" matrix. To be concrete, in our experimental setting we use the following sequence of calculations to find a distance matrix $D = \mathcal{N}(S)$ for a given symmetric similarity matrix $S$: 
\be 
D[i,j] = \sqrt{1 - \frac{S[i,j]}{\sqrt{M[i]M[j]}} }, \qquad \text{ where } M[i] \equiv \max_{j} S[i,j].
\ee 

\item \textbf{Embedder:} This is a function $\mathcal{E}$ from an $n$ by $n$ ``distance" matrix to an $n$ by $d$ ``embedding" matrix. To be concrete, in our experimental setting we use UMAP (Uniform Manifold Approximation and Projection for Dimension Reduction)~\cite{mcinnes2018umap}\footnote{\url{https://umap-learn.readthedocs.io/en/latest/}} with $d=2$.

\item \textbf{Clusterer:} This is a function $\mathcal{C}$ from an $n$ by $d$ ``embedding" matrix to a partition of $\{1,2,\ldots,n\}$. To be concrete, in our experimental setting we use HDBSCAN (Hierarchical Density-Based Spatial Clustering of Applications with Noise)\cite{mcinnes2017hdbscan}\footnote{\url{https://hdbscan.readthedocs.io/en/latest/}}. 
\end{enumerate}

\begin{algorithm}[h]
\footnotesize
\caption{Unsupervised Clustering}\label{alg:unsupervised_clustering}
\begin{algorithmic}[1]
\algrenewcommand\alglinenumber[1]{} 
\State \textbf{Input:} Trees $\{\mathcal{T}_{i}\}_{i=1}^{N}$ with labels $\{\phi_{i}\}_{i=1}^{N}$, weight function $w$, normalizer $N$, embedder $E$, clusterer $C$.
\algrenewcommand\alglinenumber[1]{\scriptsize #1}
\setcounter{ALG@line}{0}
\For{$1 \leq i < j \leq N$:}
\State Find the best matching paths $\gamma_{i},\gamma_{j}$ and score $S_{ij}$ by applying Algorithm \ref{alg:basic_match} to input trees $\mathcal{T}_{i},\mathcal{T}_{j}$, label functions $\phi_{i},\phi_{j}$ and weight function $w$.
\EndFor
\State Compute the distance matrix $D = \mathcal{N}(S)$.
\item Compute the embedding $E = \mathcal{E}(D)$. 
\State Compute and return the clustering $\mathcal{C}(E)$. 
\end{algorithmic}
\end{algorithm}

Given a cluster, we will typically want to learn the actual sequence that was planted in each tree in the cluster (either to allow us to cluster new data more efficiently, or because we believe the planted sequence is meaningful). Since Algorithm~\ref{alg:unsupervised_clustering} involves computing a best-matched sequence in each pair of trees, this amounts to solving the well-known \textit{multiple sequence alignment} (MSA) problem for this collection of best-matched sequences. The literature on this subject is too large to survey (see \textit{e.g.} \cite{ThorneKishinoFelsenstein1991MLAlignment,DoMahabhashyamBrudnoBatzoglou2005ProbCons,KatohStandley2013MAFFTv7,Gotoh1999MSAAlgorithmsApplications,msa}); in this paper we use a simple likelihood-based approach similar to \cite{ThorneKishinoFelsenstein1991MLAlignment}. We will refer to this algorithm as MSA Algorithm. 

\subsection{Matching Workflow: Features for Classifiers} 
\label{SecClassifierWorkflow}

The simplest use for path matching is as a simple feature-augmentation technique. Consider an observed dataset of labelled trees $\{\mathcal{G}_{i}\}_{i=1}^{N}$ with associated feature maps $\{\phi_{\mathcal{G}_{i}}\}_{i=1}^{N}$. We then fix a collection of templates and feature maps $\{(\mathcal{T}_{j},\phi_{\mathcal{T}_{j}})\}_{j=1}^{M}$. We can use these templates to generate auxiliary features $X_{i,j}$ by running Algorithm~\ref{alg:basic_match} with input $\mathcal{G}_{i}, \mathcal{T}_{j}, \phi_{\mathcal{G}_{i}}, \phi_{\mathcal{T}_{j}}$. These features can be used in any downstream task, \emph{e.g.}\ as input for a~classifier.

Of course, one must decide on where the pairs $\{(\mathcal{T}_{j},\phi_{\mathcal{T}_{j}})\}_{j=1}^{M}$ come from. The simplest approach for selecting $\{(\mathcal{T}_{j},\phi_{\mathcal{T}_{j}} ) \}_{j=1}^{M}$ is to randomly select paths that occur in your data - this amounts to using templates as random landmarks, as in \textit{e.g.} \cite{NIPS2007_013a006f}.

\section{Simple ``Planted Path" Model} \label{SecSynthToy}

In this section, we describe two simple data-generating process for our planted-path problem: one for randomly planting a single fixed template in a single fixed tree (Algorithm~\ref{alg:sltoc}), the other generating a collection of many trees that have hard-to-distinguish templates (Algorithm~\ref{alg:stm}). We view these models as loosely  analogous to using the stochastic block model (\textbf{SBM})~\cite{HOLLAND1983109} or a family of Artificial Benchmarks for Community Detection (\textbf{ABCD})~\cite{kaminski2021artificial} as a toy data-generating process for the clustering problem (see \emph{e.g.} ~\cite{kaminski2021mining} for more examples of using random graphs in modelling and mining complex networks). 

Before describing the process, we call attention to two pitfalls of using the \textbf{SBM} that we will try to avoid:

\begin{enumerate}
\item \textbf{Trivial Algorithms Work:} If one generates a large graph from the \textbf{SBM} with a small number of components and link probabilities chosen at random, then simply sorting nodes by degree will typically give a reasonably good clustering. We view this as a failure of the toy model, since we know that such ``trivial" algorithms typically fail on realistic data.
\item \textbf{Generated Examples Too Dissimilar to Real Data:} The usual \textbf{SBM} provides unlabelled graphs. However, in practice, node features are extremely useful for clustering. We might expect clustering algorithms developed for the \textbf{SBM} to often generalize poorly.
\end{enumerate}

We partially resolve the first problem by ensuring that, in Algorithm~\ref{alg:stm}, the distribution of the \textit{set} of labels $\{\phi(u)\}_{u \in \mathcal{T}}$ of each graph $\mathcal{T}$ is independent of the particular planted path. This ensures that trivial algorithms that do not use the tree structure, such as counting the number of times labels occur, must fail. We partially resolve the second problem by allowing a wide variety of underlying graph topologies and labels.

\subsection{Labelling Trees with Planted Paths}

Our basic model for planting a non-random sequence of features in a non-random tree is the following algorithm. We assume that we are given a tree $\mathcal{T}$ with a path $\Gamma$ of length $\ell$ identified. Our goal is to plant a random subsequence of a given sequence of features of length $k$ on a random part of the path $\Gamma$. Parameter $p$ and function $r$ model the length and, respectively, the distribution of a random subsequence. The remaining features are selected according to a distribution $\pi$. 

\begin{algorithm}[ht]
\footnotesize
\caption{Randomly Labeling Trees from Template}\label{alg:sltoc}
\begin{algorithmic}[1]
\algrenewcommand\alglinenumber[1]{} 
\State \textbf{Input:} Tree $\mathcal{T}$, connected path $\Gamma \equiv (\gamma_{1},\ldots,\gamma_{\ell}) \subseteq \mathcal{T}$, distribution $\pi$ on $\mathcal{S}$, sequence $a_{1},\ldots,a_{k} \in \mathcal{S}$, rate $r \, : \, \mathcal{S} \mapsto (0,\infty)$, observation probability $0 < p \leq 1$.

\algrenewcommand\alglinenumber[1]{\scriptsize #1}
\setcounter{ALG@line}{0}
\State Sample $N' \sim \mathrm{Bin}(\ell,p)$ and set $N = \min(N',k)$. 
\State Sample set $S_{1} \subseteq \{1,2,\ldots,\ell\}$ of size $N$ uniformly at random, and sample set $S_{2} \subseteq \{1,2,\ldots,k\}$ of size $N$ with probability proportional to $\prod_{s \in S_{2}} r(s)$. Let $s_{1}(1) < \ldots < s_{1}(N)$, $s_{2}(1) < \ldots < s_{2}(N)$ be their elements in increasing order.
\For{$i \in \{1,2,\ldots,N\}$}
\State Set $\phi(\gamma_{s_{1}(i)}) = a_{s_{2}(i)}$.
\EndFor
\For{$v \in \mathcal{T} \setminus \{\gamma_{i}\}_{i \in S_1}$}
\State Randomly set $\phi(v) \sim \pi$.
\EndFor
\State Return $\phi$.
\end{algorithmic}
\end{algorithm}

In a typical dataset, one might have a family of trees, each with one of many possible planted paths. Before describing an associated synthetic random model, we need some further notation. We say that $\mu$ is a \textit{distribution on (tree, path) pairs} if a sample $(\mathcal{T}, \Gamma) \sim \mu$ consists of a directed tree $\mathcal{T}$ and a sequence $\Gamma = \{\gamma_{1},\ldots,\gamma_{k}) \subseteq \mathcal{T}$ that satisfies $\gamma_{i} <_{\mathcal{T}} \gamma_{i+1}$ for all $i \in \{1,2,\ldots,k-1\}$. 

For our experiments, we will use a simple model that fits our data reasonably well. Our model generates random trees $\mathcal{GW}(N,\lambda)$ with depth at most $N$ and aims to produce ``bushy'' trees. Fix a positive integer $N$ and a real number $\lambda > 1$. To generate a random tree, we will use a well-known Galton-Watson process in which the number of children of nodes are i.i.d.\ random variables that follow the Poisson distribution with mean $\lambda$. 

\begin{algorithm}[ht]
\footnotesize
\caption{Sampling from Toy Model}\label{alg:stm}
\begin{algorithmic}[1]
\algrenewcommand\alglinenumber[1]{} 
\State \textbf{Input:} Distribution $\mu$ on (tree, path) pairs, distribution $\pi$ on $\mathcal{S}$, base sequence $a_{1},\ldots,a_{k} \in \mathcal{S}$, rate $r \, : \, \mathcal{S} \mapsto (0,\infty)$, observation probability $0 < p \leq 1$, number of observations per class $n_{1},\ldots,n_{M}$. 

\algrenewcommand\alglinenumber[1]{\scriptsize #1}
\setcounter{ALG@line}{0}
\For{$i \in \{1,2,\ldots,M\}$:}
\State Sample permutation $\sigma_{i}$ of $\{1, 2, \ldots, k\}$ uniformly at random.
\For{$j \in \{1,2,\ldots,n_{i}\}$:}
\State Sample a tree and path $(\mathcal{T}_{ij},\Gamma_{ij}) \sim \mu$.
\State Sample labels $\phi_{ij}$ according to Algorithm \ref{alg:sltoc} with input tree $\mathcal{T}_{ij}$, path $\Gamma_{ij}$, distribution $\pi$, sequence $a_{\sigma_{i}(1)},\ldots,a_{\sigma_{i}(k)}$, rate $r$, and observation probability $p$.
\EndFor
\EndFor

\State Return $\pi, \mathcal{T}_{ij}, \Gamma_{ij}, \phi_{ij}$.
\end{algorithmic}
\end{algorithm}

In our basic model, Algorithm~\ref{alg:stm}, all of the planted sequences $\{a_{\sigma_{i}(j)}\}_{j=1}^{k}$ are permutations of a single base sequence $\{ a_{j}\}_{j=1}^{k}$. This has the advantage of ensuring that the distributions of sets of labels are completely independent of the class $i \in \{1,2,\ldots, M\}$, ensuring that any algorithms based only on the sets of labels cannot distinguish between the classes. Of course, one could make a small tweak to this algorithm to allow any collection of planted sequences.

Algorithms~\ref{alg:sltoc} and~\ref{alg:stm} deal with two extreme cases: either both tree and planted path are fixed, or neither are fixed. Of course, it is reasonable to consider the intermediate situation for which the planted path is known and fixed while the trees are random. This corresponds to running Algorithm \ref{alg:stm} with $M=1$.

\subsection{Sanity Check: Is This Problem Nontrivial but Tractable?} 
\label{sec:GW-sanity}

A natural question is: for our simulated data, how hard is it to distinguish a random tree \textit{with} a planted path from a random tree \textit{without} the path? Our goal is to check that this problem is nontrivial (in that simple approaches, such as label-counting, fail) and tractable (in that the problem would be fairly easy if you knew the path itself).

For this experiment, we use Algorithm \ref{alg:stm} (with base tree generator described in Section~\ref{sec:GW-tree}) to generate 200 trees from $M=2$ classes and observation probability $p=0.75$, with a very small base alphabet $|\mathcal{S}|=5$. Let us highlight the fact that a small alphabet with ``bushy'' trees ensures that there will be many ``accidental" or ``false positive" matches. Our scoring algorithm would look substantially better, even for much smaller values of $p$, if $|\mathcal{S}|$ were much larger. 

The left-hand side of Figure~\ref{fig:sim_hist} shows a typical tree, with the candidate planted path highlighted (this is before subsampling at rate $p=0.75$). The ``similarity score" plotted on the right-hand side of Figure~\ref{fig:sim_hist} is computed by finding the similarity of each tree and the reference sequence $a$ associated with class 1 using Algorithm~\ref{alg:basic_match}. We then plot separately the similarities of all trees in class 1 (``within-cluster") and those in class 2 (``outside-cluster"). This experiment shows that our similarity score is very good at distinguishing between these two classes even when the planted path is both (i) a small fraction of the full ``bushy" tree and (ii) sampled at a fairly small rate of $p=0.75$.

\begin{figure}[ht]
    \centering
    \includegraphics[width=0.8\linewidth]{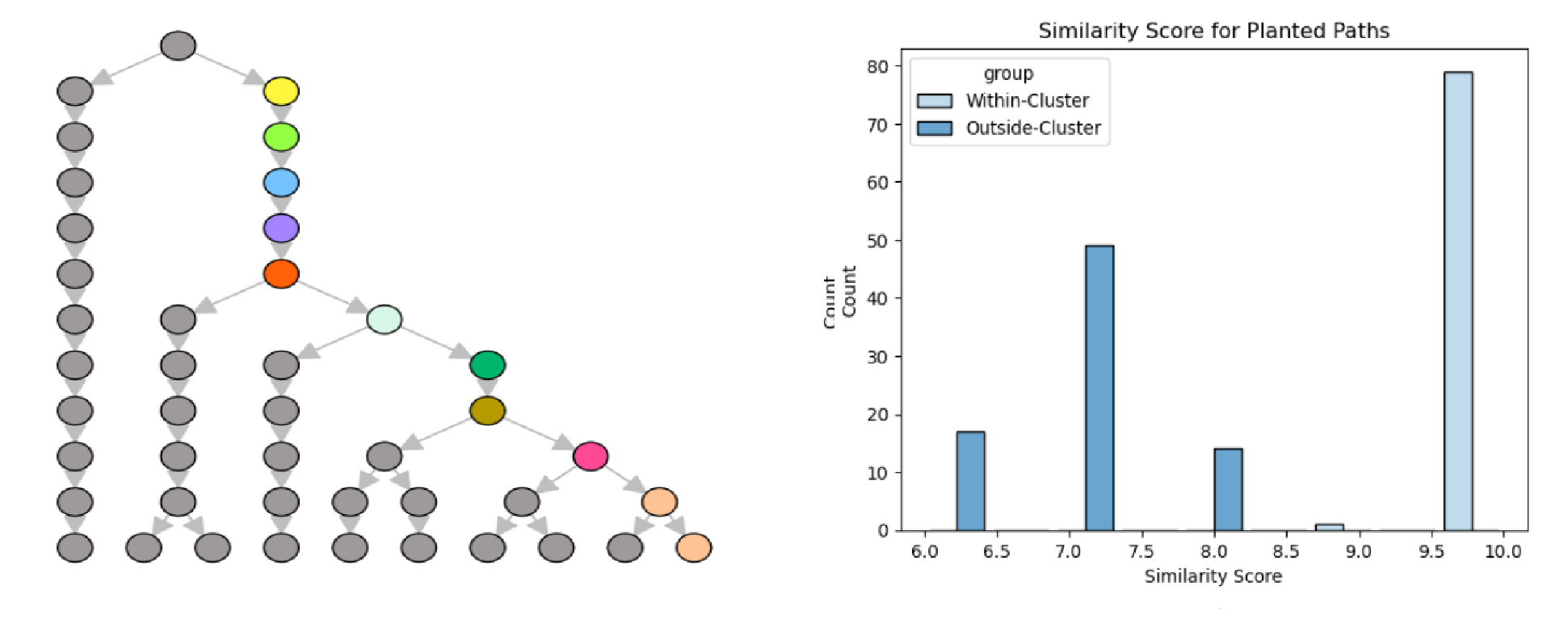}
    \caption{Example of a typical tree (left) and histogram of similarity scores for the two classes (right).}
    \label{fig:sim_hist}
\end{figure}

Figure~\ref{fig:count_hist}, on the other hand, shows that the histograms of symbol frequencies are indistinguishable. In these figures, for each tree, we compute the number of times a given symbol appears as a node label. We then plot the histogram of these counts for within- and outside-cluster trees. 

\begin{figure}[ht]
    \centering
\includegraphics[width=0.9\linewidth]{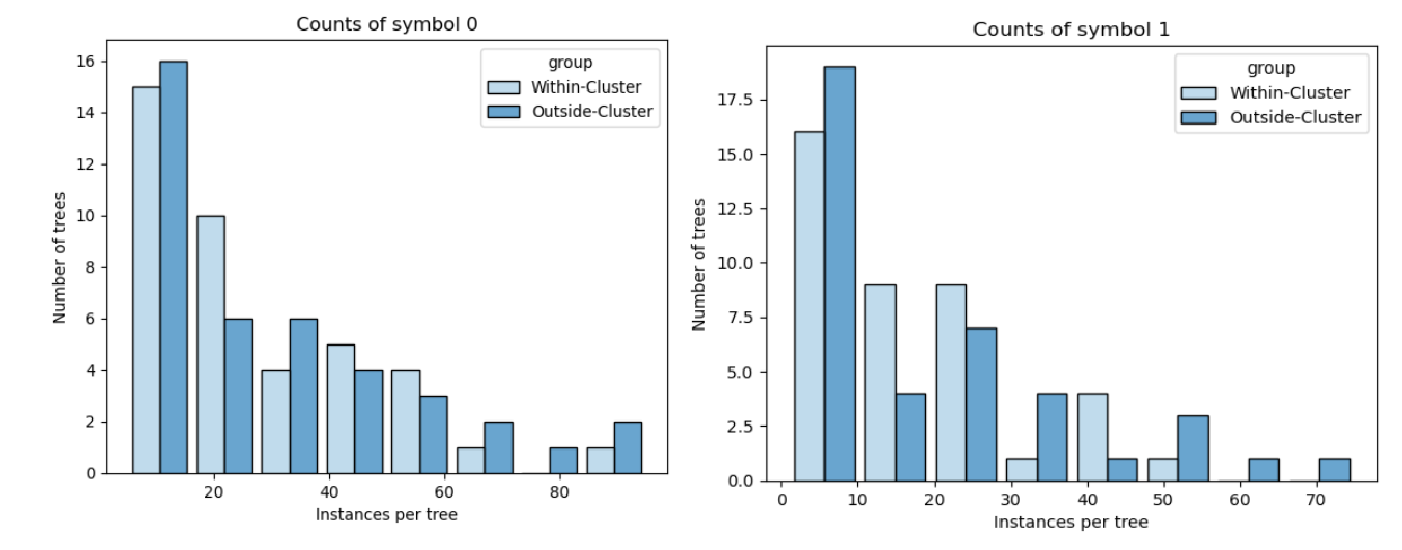}
    \caption{Histograms of symbol frequencies.}
    \label{fig:count_hist}
\end{figure}

\subsection{Experiments on the Synthetic Model}

We check that the unsupervised workflow proposed in Section~\ref{SecWorkUns} works reasonably well for this toy ``planted path" model.
In our first experiment, we generate labelled trees using Algorithm~\ref{alg:stm}, with the main parameters being:
\begin{enumerate}
\item The base distribution on trees is $\mathcal{GW}(12,1.8)$,
\item The base distribution $\pi$ on the alphabet is uniform on $\mathcal{S} = \{A,B,C,D,E\}$, 
\item The base sequence $a$ is sampled from length-10 sequences using $\pi$,
\item The observation probability is $p=0.9$, with rate $r$ constant, and
\item The number of observations per class is $n_{1} = \ldots = n_{4} = 50$.
\end{enumerate}

We then run Algorithm~\ref{alg:unsupervised_clustering}, with parameters as described in Section~\ref{SecWorkUns}. The resulting embedding is presented in Figure~\ref{fig:cluster_pic} (left). 
Additionally, we use our variant of the MSA Algorithm to extract exemplars, one for each family, and see whether they separate classes; see Figure~\ref{fig:cluster_pic} (right).
Even a quick visual inspection shows that both algorithms do fairly well, with the four ground-truth clusters well-separated and distance-to-exemplars much smaller for trees in the same class than trees in other classes. 

\begin{figure}[ht]
    \centering
    \includegraphics[width=0.9\linewidth]{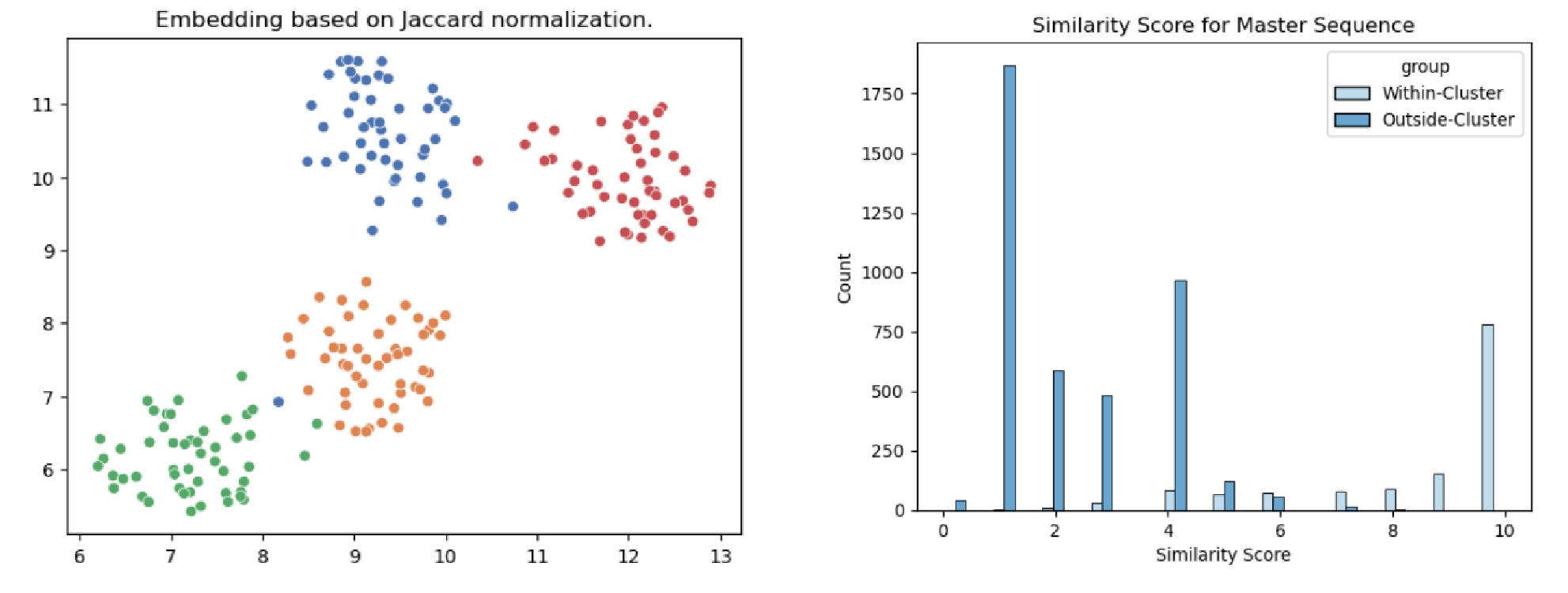}
    \caption{Embedding of trees coloured by the ``ground-truth" cluster label (left). Similarity score for the extracted exemplars (right).}
    \label{fig:cluster_pic}
\end{figure}

Next, we consider essentially the same experiment, with four changes: we have $n_{1} = \ldots = n_{6} = 18$ observations per class and $p=0.7$ (so that the base problem becomes slightly harder), we set $\pi(i) \approx (0.247, 0.247, 0.247, 0.247, 0.012)$ (that is, the symbol $E$ is now rare), and we put two copies of $E$ inside the base sequence $a$. We then run Algorithm~\ref{alg:unsupervised_clustering} as above, but with two different choices for $w$: the ``unweighted" choice $w \equiv 1$, and the ``weighted" choice given in formula \eqref{SimpleProbScoring}. The two resulting embeddings are displayed in Figure~\ref{fig:weighted_cluster_pic}.
A quick visual inspection of Figure~\ref{fig:weighted_cluster_pic} shows that a probability-weighted score function substantially improves separation of the clusters, as expected.

\begin{figure}[ht]
    \centering
    \includegraphics[width=0.9\linewidth]{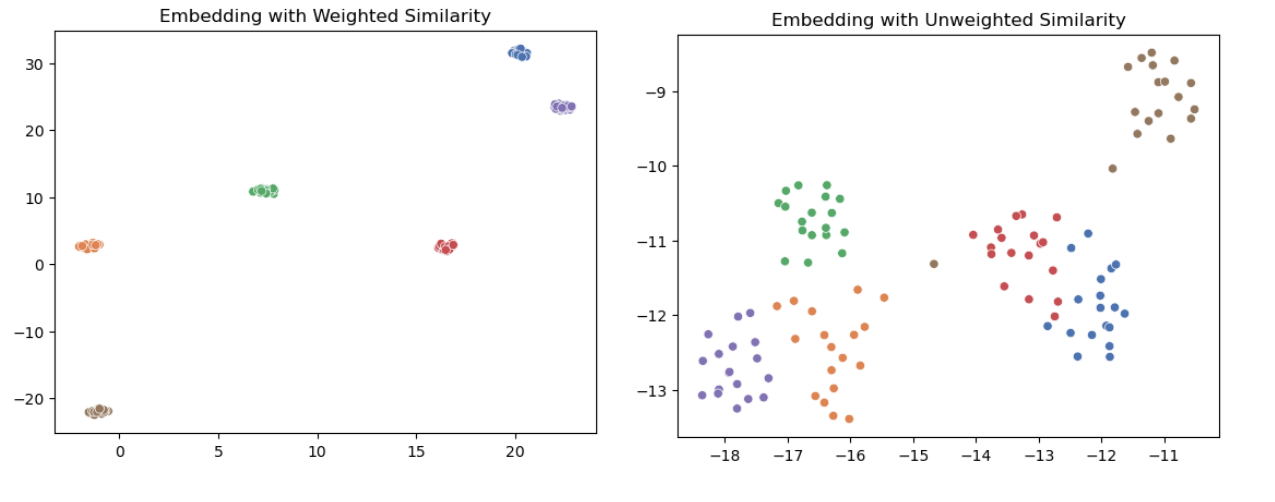}
    \caption{Two embeddings of six classes of trees, coloured by the ``ground-truth" cluster label: weighted (left) and ``unweighted'' (right).}
    \label{fig:weighted_cluster_pic}
\end{figure}

\section{Matching Paths in the ACME4 Dataset}\label{SubsecACMEExpts}

\textbf{ACME} refers to cybersecurity datasets using the open-source Wintap telemetry collection tool developed by the Lawrence Livermore National Laboratory (LLNL). Designed specifically for cybersecurity research, \textbf{ACME} addresses the complexities of gathering and analyzing host-based data within large-scale Windows environments. Further documentation for both Wintap and \textbf{ACME} is available on-line\footnote{\url{https://gdo168.llnl.gov}}. 

We use the \textbf{ACME4} dataset, in particular, the tables defined in the standard view collection\footnote{\url{https://gdo168.llnl.gov/data/newdocs/datadict}}. The \textbf{ACME4} data simulates a Windows business network with 10 workstations over a 2-week period and a variety of attacks (Living off the Land, Caldera, Metasploit, etc.). Simulated malicious actors are explicitly identified via ``bad'' username labels.

\subsection{Process Trees}

To build process trees, we extracted the following information from the {\tt process\_uber\_table} in the standard view collection for every process:
the process identifier (PID) hash, its parent PID hash, the process name (often blank), and the user name (also often blank).
From this data, we built a large directed graph where each edge is directed from parent to child. We then extracted the connected components, which correspond to the process trees.
We obtained 1,111,277 trees, for a total of 2,884,171 nodes (PIDs) and 1,772,894 edges.
The sizes of the trees range from 2 to 257,744, with the vast majority being of size~2.

Looking at every process tree, only 8 contain bad users. We show the smallest such tree in Figure~\ref{fig:bad_tree}, where each node is indexed with a shortened and anonymized user name (such as ``user88'', ``bad3'' or ``SYS'') followed by the process name and separated with ``::''; recall that both values can be blank.

\begin{figure}[!ht]
    \centering
    \vspace{-.75cm}
    \includegraphics[width=0.9\linewidth]{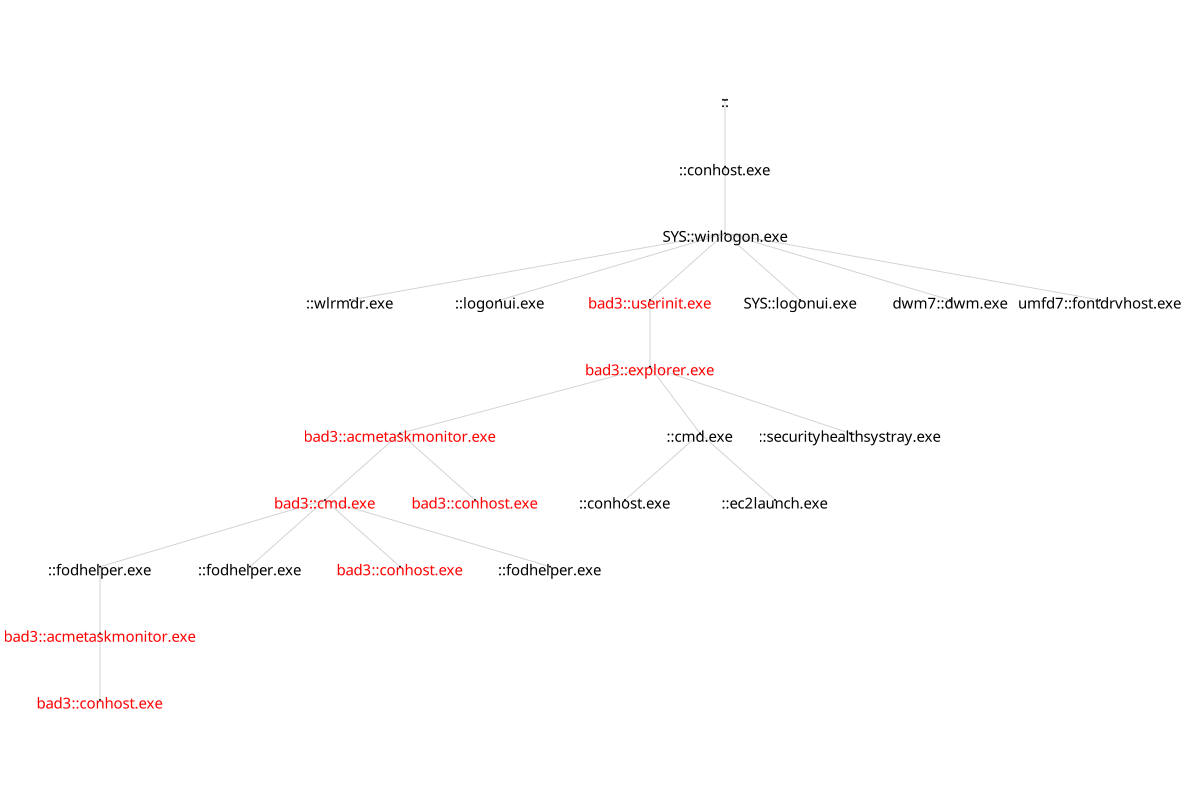}
    \caption{Example of a process tree containing PIDs associated with ``bad'' users (shown in red). We use {\tt short\_username::process\_name} as the node labels.}
    \label{fig:bad_tree}
\end{figure}

\subsection{Matching Algorithm Experiments}

In order to illustrate the use of matching algorithms such as Algorithm 1, we consider paths consisting of more that a single edge.
Since the vast majority of process trees in the ACME4 dataset are of size 2, we build a larger set of sub-trees as follows.

For every process tree of size 3 to 10,000 (there are 210 of those), we consider every sub-tree of size 3 or more for which the root process name is not blank or labelled as ``unknown''. This yields 2,693 process sub-trees, with sizes ranging from 3 to 9,558 nodes, and depth from 2 to 15.
Of those, 131 trees have at least one process associated with a ``bad'' user name. We selected one such tree as our ``reference'' tree, shown in Figure \ref{fig:algo1_example}(a). We then applied Algorithm~\ref{alg:basic_match}, looking for matching path(s) with high scores in all other process trees. High scoring trees are shown in Figure \ref{fig:algo1_example} using three different scoring functions to compare the node labels. 
In Figure \ref{fig:algo1_example}(b), we consider only the process names with a binary score (1 if the processes are the same, and 0 otherwise). We see that we found a matching path consisting of 6 nodes with the same process names as the path highlighted in (a), thus with a total score of 6, but the path in (b) has blank usernames. Note that we show all nodes not part of the matching path as simple dots for easier visualization.
In Figure \ref{fig:algo1_example}(c), we consider both the process names and usernames, where all usernames starting with ``user'' or ``bad'' are considered as matching. We use a binary function equal to 1 only if both the processes and usernames match. We again found a matching path consisting of 6 nodes (with a total score of 6) with the same process names as the path highlighted in (a), but the path in (c) has username ``user1'' instead of ``bad3'' (and the two other usernames also match, namely blank and SYS).
In Figure \ref{fig:algo1_example}(d), we do as in (c), but we give partial score for each matching part: 0.75 if the process names match, and 0.25 if the usernames match. We show a matching path with a total score of 5.75, due to the fact that 3 of the 4 instances of ``bad3'' username in (a) correspond to the non-blank username ``user88'' in (d), but the fourth username is blank.

\begin{figure}[ht]
    \centering    
    \vspace{-.25cm}
    \begin{subfigure}{.49\textwidth}
        \includegraphics[width=\textwidth]{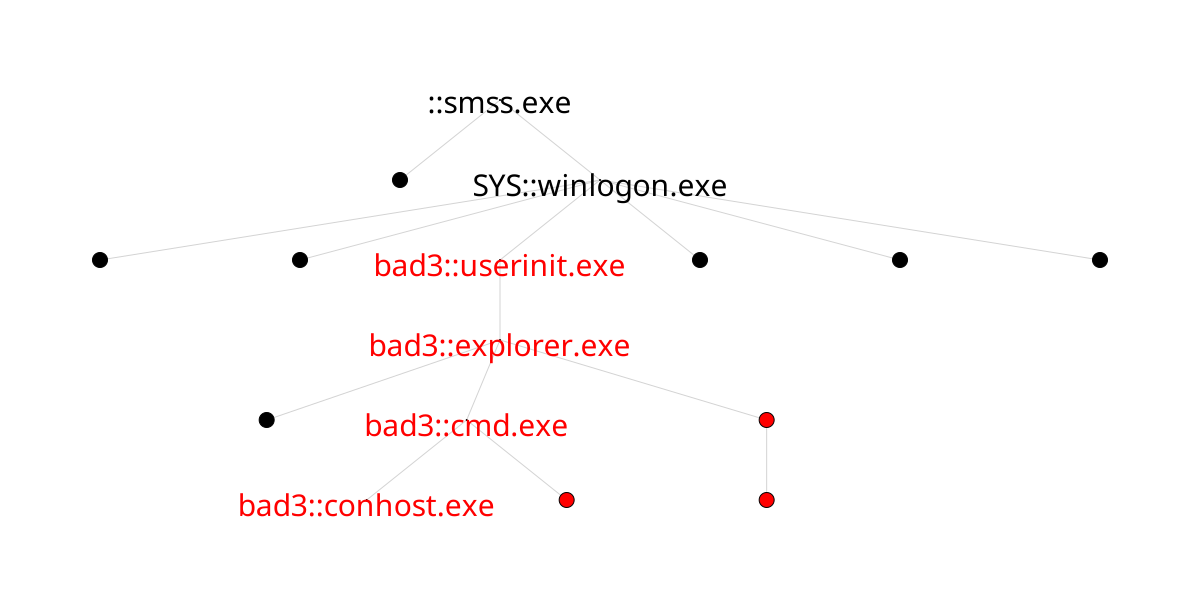}
        \caption{Reference tree}
        \label{fig:1}
    \end{subfigure}
    \hfill
    \begin{subfigure}{0.49\textwidth}
        \includegraphics[width=\textwidth]{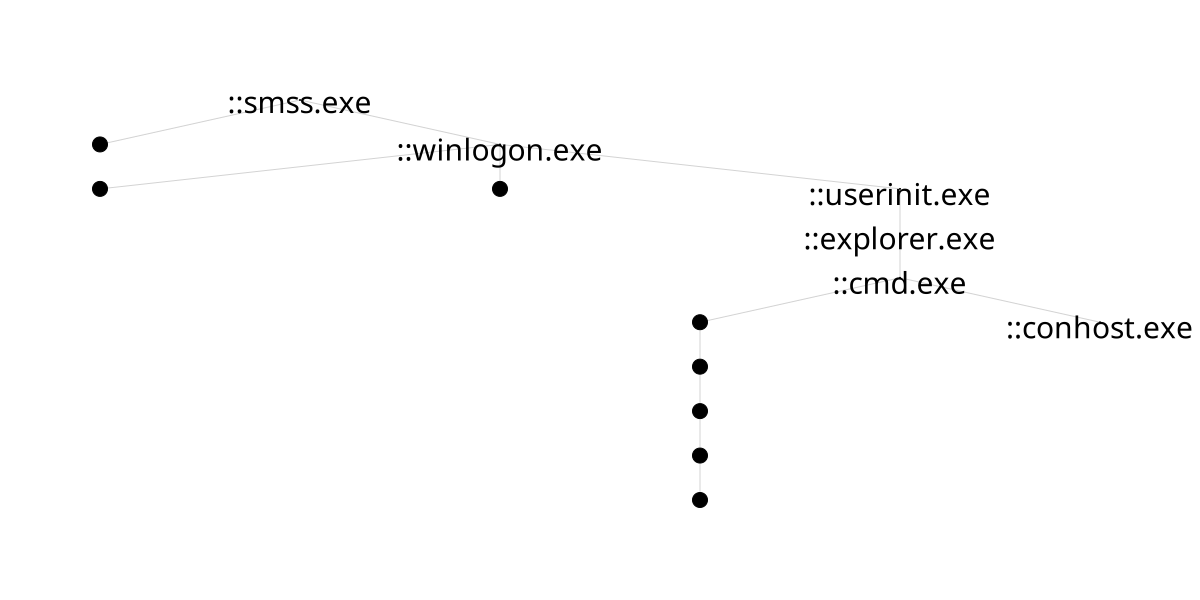}
        \caption{Exact matching on processes}
        \label{fig:2}
    \end{subfigure}

    \begin{subfigure}{.49\textwidth}
        \includegraphics[width=\textwidth]{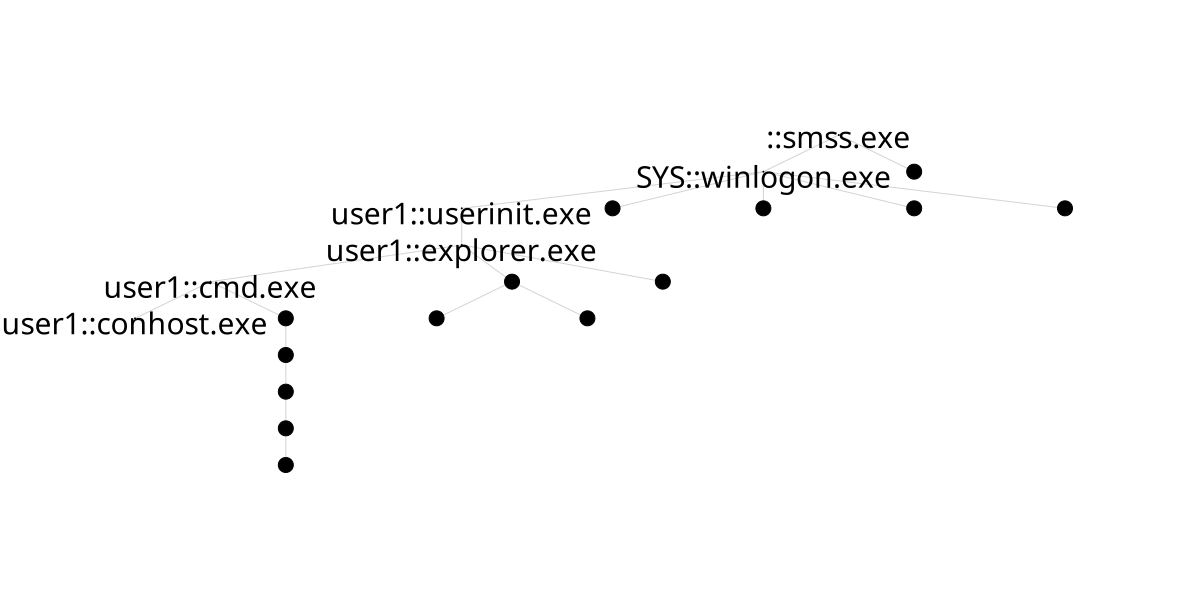}
        \caption{Exact matching on processes and users}
        \label{fig:3}
    \end{subfigure}
    \hfill
    \begin{subfigure}{0.49\textwidth}
        \includegraphics[width=\textwidth]{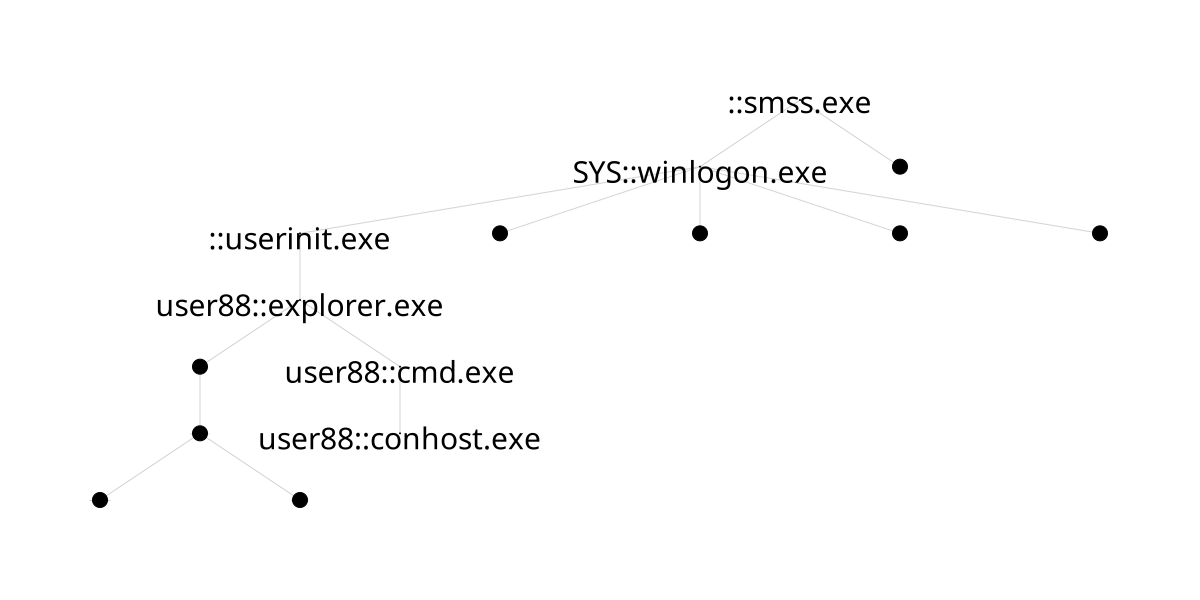}
        \caption{Soft matching on processes and users}
        \label{fig:4}
    \end{subfigure}

\caption{Results from Basic Matching Algorithm (Algorithm 1) using different similarity scores for the node labels. Node labels are shown for the matching paths only.}
    \label{fig:algo1_example}
\end{figure}

The results shown in Figure \ref{fig:algo1_example} can be seen as building blocks for more complex workflows, such as the ones described in Section \ref{SecClassifierWorkflow}. 
We illustrate two such workflows.

For the first illustration, we use a subset of the process trees we just described, keeping the ones with depth at least 4, and at most 100 nodes. There are 511 such trees, 38 of which contain at least one instance of a bad user, which we use as our templates ${\mathcal{T}_{j}}$. 
For the other 473 process trees ${\mathcal{G}_{i}}$, we run our matching algorithm against every template and store the resulting scores (length of the longest common path) as features $X_{i,j}$. We then use those features to score the process trees as follows. For each $\mathcal{G}_{i}$, we compute $S_i = \sum_j \mathbf{1}_{\{X_{i,j} \ge 3\}}$, thus counting the number of templates having a matched path with 3 nodes or more. We plot the resulting score distribution in Figure \ref{fig:acme_classify}(a), where we see that a large number of trees have very low score, while a smaller group show a large score. This is an example where matching algorithm can be used as a filter, in this case selecting a subset of process trees having paths in common with several template trees. 
Looking more closely at the results, one clear difference is that trees ${\mathcal{G}_{i}}$ with score $S_i \ge 9$ tend to be larger than other trees, with respective median values of 22 and 6 nodes.
Looking at the matching paths between templates and high scoring trees, the most frequent sequence of processes is {\tt winlogon.exe-userinit.exe-explorer.exe-cmd.exe-conhost.exe} or a subset thereof, which coincides with the results shown in Figure \ref{fig:algo1_example}.

Four process names frequently found at the root of our process trees are: {\tt bash.exe, cmd.exe, smss.exe} and {\tt taskhostw.exe}.
For the second illustration, we use the same 38 templates as before and we select a subset of 832 process trees ${\mathcal{G}_{i}}$, distinct from the templates, where the root process is one of those 4 processes.
For each process tree,
we run our matching algorithm against each template and store the resulting scores (length of the longest common path) as a 38-long feature vector $\{X_{i,j}, 0 \le j \le 37\}$. 
Randomly selecting half of the
${\mathcal{G}_{i}}$'s for training and the other half for testing,
we trained a simple random forest classifier using the 4 process names as labels.
The resulting confusion matrix is shown in Figure \ref{fig:acme_classify}(b), where we see that we can correctly classify most trees, with the exception of a small group of trees with root label {\tt bash.exe} being classified as {\tt cmd.exe}.

\begin{figure}[h]
    \centering 
    \begin{subfigure}{.46\textwidth}
        \includegraphics[width=\textwidth]{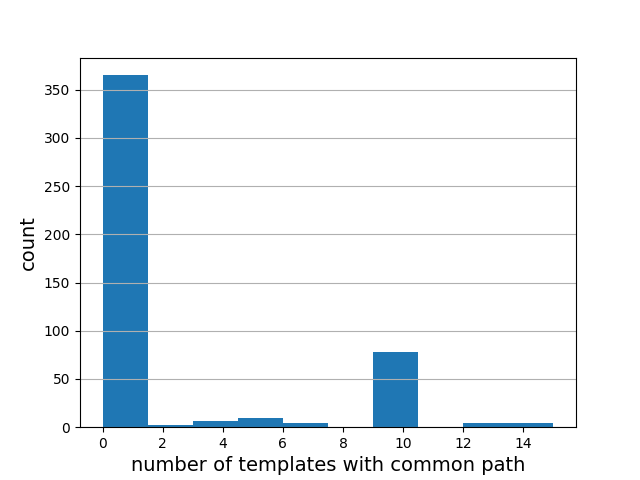}
        \caption{}
        \label{fig:one}
    \end{subfigure}
    \hfill
    \begin{subfigure}{0.47\textwidth}
        \includegraphics[width=\textwidth]{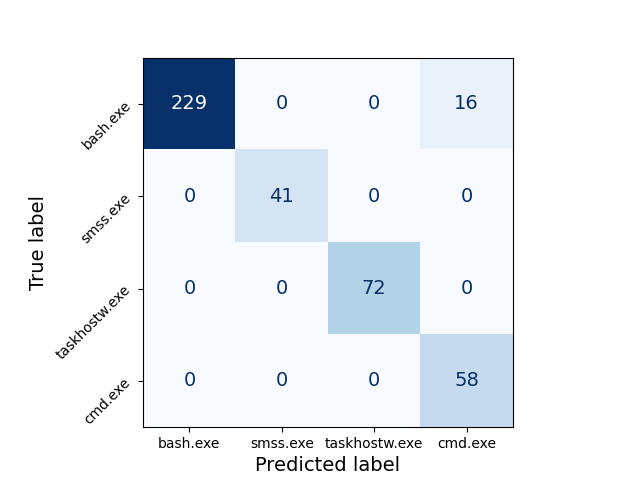}
        \caption{}
        \label{fig:two}
    \end{subfigure}
\caption{(a) Counts of the number of process trees vs number of templates with a common path. (b) Confusion matrix for a random forest classifier using our matching algorithm to build the features.}
    \label{fig:acme_classify}
\end{figure}

\bibliography{ESCBib}

\appendix


\section{Correctness of Algorithm \ref{alg:basic_match}}

We claim that Algorithm~\ref{alg:basic_match} gives the desired optimal path:

\begin{theorem} \label{thm_basic_correctness}
Fix $\mathcal{G},\mathcal{H}, \phi_{\mathcal{G}},\phi_{\mathcal{H}},w$. Let $a_{1:L}, A$ be the output of Algorithm~\ref{alg:basic_match} with this input. Then $a_{1:L}$ is a valid matching, and the score of $a_{1:L}$ satisfies
\be 
s(a_{1:L}) = A = s(\mathcal{G},\mathcal{H}).
\ee 
\end{theorem}

We break the proof into two lemmas, corresponding to the ``forward" and ``backward" passes of Algorithm~\ref{alg:basic_match}. We state both lemmas before we prove them.

\begin{lemma} [Score Correctness] \label{lemma_basic_correctness_score}
Fix notation as in Theorem~\ref{thm_basic_correctness}. Let $\{A_{u,v}\}_{u \in [n], v \in [m]}$ be the matrix computed in Algorithm~\ref{alg:basic_match}. Then, for all $u \in [n], v \in [m]$,
\begin{eqnarray} \label{EqCorrectness}
A_{u,v} = \max_{p \in \mathcal{P}_{u,v}} s(p).
\end{eqnarray}
In particular, this implies the final output $A$ satisfies $A = \max_{p \in \mathcal{P}} s(p).$
\end{lemma}

\begin{lemma} [Path Correctness] \label{lemma_basic_correctness_path}
Fix notation as in Theorem \ref{thm_basic_correctness}. Then $a_{1:L}$ is a valid path, and
\be 
A = s(a_{1:L}).
\ee 
\end{lemma}

We prove the two lemmas.

\begin{proof} [Proof of Lemma \ref{lemma_basic_correctness_score}]
We prove Equation \eqref{EqCorrectness} by ``double induction" on the indices $u,v$. For the base case, the values $A_{u,v}$ with either $u=-1$ or $v=-1$ correspond to empty matchings, and so these trivially satisfy Equation \eqref{EqCorrectness}. 

We now consider the formula for $A_{u,v}$ in line 4 of Algorithm \ref{alg:basic_match}, and assume that $A_{u',v'}$ satisfies Equation \eqref{EqCorrectness} for all previously-computed $u',v'$. Let $q \in \mathcal{P}_{u,v}$ be a valid partial matching that satisfies $s(q) = \max_{p \in \mathcal{P}_{u,v}} s(p)$. Then consider the last element $(q_{L}(1), q_{L}(2))$ of the list $q$. We have three cases, which we consider in the same order as the three cases in line 4 of Algorithm \ref{alg:basic_match}:

\begin{enumerate}
\item \textbf{Case 1: $q_{L}(1) \neq u$.} In this case, $q \in \mathcal{P}_{anc_{\mathcal{G}}(u), v}$, so by the induction hypothesis 
\be 
A_{u,v} = s(q) \leq A_{anc_{\mathcal{G}}(u), v}.
\ee 
Since $q$ maximizes the score amongst paths in $\mathcal{P}_{u v}$ and $\mathcal{P}_{anc_{\mathcal{G}}(u), v} \subseteq \mathcal{P}_{u v}$, we must in fact have
\be 
A_{u,v} =  s(q) = A_{anc_{\mathcal{G}}(u), v}.
\ee 
\item \textbf{Case 2: $q_{L}(2) \neq v$.} By the same argument as in Case~1, 
\be 
A_{u,v} = s(q) = A_{u,anc_{\mathcal{H}}(v)}.
\ee 
\item \textbf{Case 3: $(q_{L}(1), q_{L}(2)) = (u,v)$.} In this case, the list $q_{1},\ldots,q_{L-1} \in \mathcal{P}_{anc_{\mathcal{G}}(u),anc_{\mathcal{G}}(v)}$. Thus, 
\begin{eqnarray*} 
A_{u,v} &= s(q) \\
&= w(\phi_{\mathcal{G}}(u),\phi_{\mathcal{H}}(v)) + s((q_{1},\ldots,q_{L-1}))\\
&\leq w(\phi_{\mathcal{G}}(u),\phi_{\mathcal{H}}(v)) + A_{anc_{\mathcal{G}}(u), anc_{\mathcal{H}}(v)}, 
\end{eqnarray*}
where the third line is due to the induction hypothesis. On the other hand, by choosing $q' \in \mathcal{P}_{anc_{\mathcal{G}}(u),anc_{\mathcal{G}}(v))}$ to satisfy $s(q') = \max_{p \in \mathcal{P}_{anc_{\mathcal{G}}(u),anc_{\mathcal{G}}(v))}} s(p)$, we can see that the path $q''$ obtained by appending $(u,v)$ to $q'$ is in $\mathcal{P}_{u,v}$ and satisfies
\be 
s(q') = w(\phi_{\mathcal{G}}(u),\phi_{\mathcal{H}}(v)) + A_{anc_{\mathcal{G}}(u), anc_{\mathcal{H}}(v)}.
\ee 
In fact, since $q$ had the highest possible score amongst paths in $A_{u,v}$, the truncation to $(q_{1},\ldots,q_{L-1})$ must have had the highest possible score amongst paths in $A_{anc_{\mathcal{G}}(u), anc_{\mathcal{H}}(v)}$. Thus, in fact 
\be 
A_{u,v} = s(q) =  w(\phi_{\mathcal{G}}(u),\phi_{\mathcal{H}}(v)) + A_{anc_{\mathcal{G}}(u), anc_{\mathcal{H}}(v)}.
\ee 

\end{enumerate}
Thus, in all three cases, we matched $A_{u,v}$ to one of the corresponding expressions on the right-hand side of line 5. This implies that $A_{u,v}$ as calculated in that line must have the correct value.

The last claim is obvious since $\mathcal{P} = \bigcup_{u,v} \mathcal{P}_{u,v}$, and we have just seen that $A_{u,v}$ contains the highest score within each of the component sets $\mathcal{P}_{u,v}$.
\end{proof}

Since Lemma \ref{lemma_basic_correctness_path} is obvious but the proof is tedious to formalize, we only provide a sketch. 

\begin{proof} [Sketch of the Proof of Lemma \ref{lemma_basic_correctness_path}]
Once a dynamic programming algorithm completes a bottom-up tabulation (the ``forward pass''), one ends up with a single value that optimizes a given problem at hand. In our particular case, it is the maximum score over all valid matchings. To find the actual matching that led to that result, the algorithm needs to perform a ``backward pass.'' This is the tabulation $C_{u,v}$ equivalent of backtracking: we are reversing the decisions made during the filling of the matrix $A_{u,v}$.
\end{proof}

\section{Algorithm \ref{alg:basic_match} --- Small Variations and Incorporating Prior Knowledge} \label{SubsecSmallVar}

When Algorithm~\ref{alg:basic_match} is applied with score $\mathcal{S} \equiv 1$, it returns the largest common sub-path of the two trees. It can be viewed as a reasonable ``default" behaviour for a planted-path algorithm, but in practice one may want slightly different behaviour. We briefly sketch some common adjustments but full details are left for the journal version of this paper.

\begin{itemize}
    \item \textbf{Nonuniform match importance:} Some matches may be more informative than others. A standard default choice is to ``fit" the score function $w$ to the data, setting
    \be 
    w(i,i) \propto \frac{2}{N(i)+2},
    \ee where $N(i)$ is the number of times label $i$ appears in the dataset (see \textit{e.g.} Chapter~8 of \cite{Manning_Raghavan_Schutze_2008} for an introduction to scoring rules of this nature). 
    \item \textbf{Subtrees rather than paths:} We may want to find the optimal matching of two full \textit{subtrees} rather than \textit{paths}. We can allow for this by replacing the condition 
\[
u_{i} <_{\mathcal{G}} u_{i+1} \text{ and } v_{i} <_{\mathcal{H}} v_{i+1} \text{ for all } i \in \{1,2,\ldots, k-1\}
\]
in Definition~\ref{DefValidPartialMatchingPaths} by the relaxed condition 
\[
u_{i} <_{\mathcal{G}} u_{j} \text{ if and only if } v_{i} <_{\mathcal{H}} v_{j} \text{ for all } i,j \in \{1,2,\ldots, k-1\}.
\]
The ``optimal score" and recursive algorithms can be easily tweaked.
\item \textbf{Preference for shorter matches: } Increasing the lengths of matches always increases our score, and so our method is biased towards giving long matches (even if some parts of the match contribute very little to the score). We can find the best score for paths \textit{of each length} by replacing  $A_{u,v}$ with
\be \label{EqTweakShort}
A_{u,v,r} = \max_{p \in \mathcal{P}_{u,v,r}} s(p),
\ee 
where 
\be 
\mathcal{P}_{u,v,r} = \{ \gamma \in \mathcal{P}_{u,v} \, : \, | \gamma | = r\}.
\ee 
\item \textbf{Small-gap, high-density, or top-$K$:} Just as we sometimes wish to find small paths, we might want to find paths with a small number of ``skipped" nodes. We might also wish to find the ``top-$K$" paths, rather than the single best match. These can be computed by making essentially the same tweak as in Equation~\eqref{EqTweakShort}.
\end{itemize}

\section{Generating Random Trees and Random Paths} \label{sec:GW-tree}

Algorithm~\ref{alg:stm} requires an algorithm for generating (tree, path) pairs. We give a simple example that fits our data reasonably well and which we use for simulation experiments. Our model generates random trees with depth at most $N$ and aims to produce ``bushy'' trees. Fix a positive integer $N$ and a real number $\lambda > 1$. To generate a random tree, we will use a well-known Galton-Watson process in which the number of children of nodes are i.i.d.\ random variables that follow the Poisson distribution with mean $\lambda$. 

Let $X_k$ be the number of nodes at level $k$. Then, the sequence $(X_k)$ evolves according to the recurrence formula: $X_0 = 1$ and for any $k \in \mathbb{N}$, $X_k$ is a sum of $X_{k-1}$ i.i.d.\ random variables $\mathrm{Po}(\lambda)$. We stop the process when either $X_{k} =0$ (there are no nodes at level $k$) or $k = N$ (we have reached the maximum depth). Note that this procedure might not succeed as the process may finish before getting to depth $N$. If this happens, then we simply re-start the process from the beginning and continue doing so until it eventually generates the desired tree. Let $\mathcal{GW}(N,\lambda)$ be the probability distribution over the generated trees of depth $N$.

\begin{remark} [Complexity of Rejection Sampling]
It is natural to ask: how long does it take to sample from $\mathcal{GW}(N,\lambda)$? It is known that the probability of the Galton-Watson process ``going on forever'' (the survival probability) is $1-q$, where $q = q(\lambda)$ (the probability of eventual extinction) satisfies the following equality: $q = e^{\lambda(q-1)}$. Hence, each independent attempt to generate a tree succeeds with probability at least $1-q$. For example, $1-q \approx 0.797$ when $\lambda = 2$ and $1-q \approx 0.940$ when $\lambda = 3$. For such values of $\lambda$ the expected number of re-samplings is very small, regardless of the value of $N$.
\end{remark}

For a given tree $\mathcal{T}$ with a root $r$, regardless whether generated at random or not, we may simply select a random directed path as follows. Let $u$ be one of the leaves of $\mathcal{T}$ selected uniformly at random. Then, let $\Gamma \equiv (r = \gamma_{1}, \gamma_{2}, \ldots, u=\gamma_{\ell}) \subseteq \mathcal{T}$ be the unique path from $r$ to $u$ in $\mathcal{T}$.

Figure \ref{fig:bushy} illustrates the results of sampling trees and paths using this process, then labeling them according to Algorithm \ref{alg:sltoc}. The left-hand tree corresponds to $\lambda = 1.24$, $N=11$; the right-hand tree corresponds to $\lambda=2.22$, $N=4$. In both cases the same base sequence $a$ and matching probability $p=0.99$ were used in Algorithm~\ref{alg:sltoc}, with the planted path shown in bright colours.

\begin{figure}[H]
    \centering
    \includegraphics[width=0.8\linewidth]{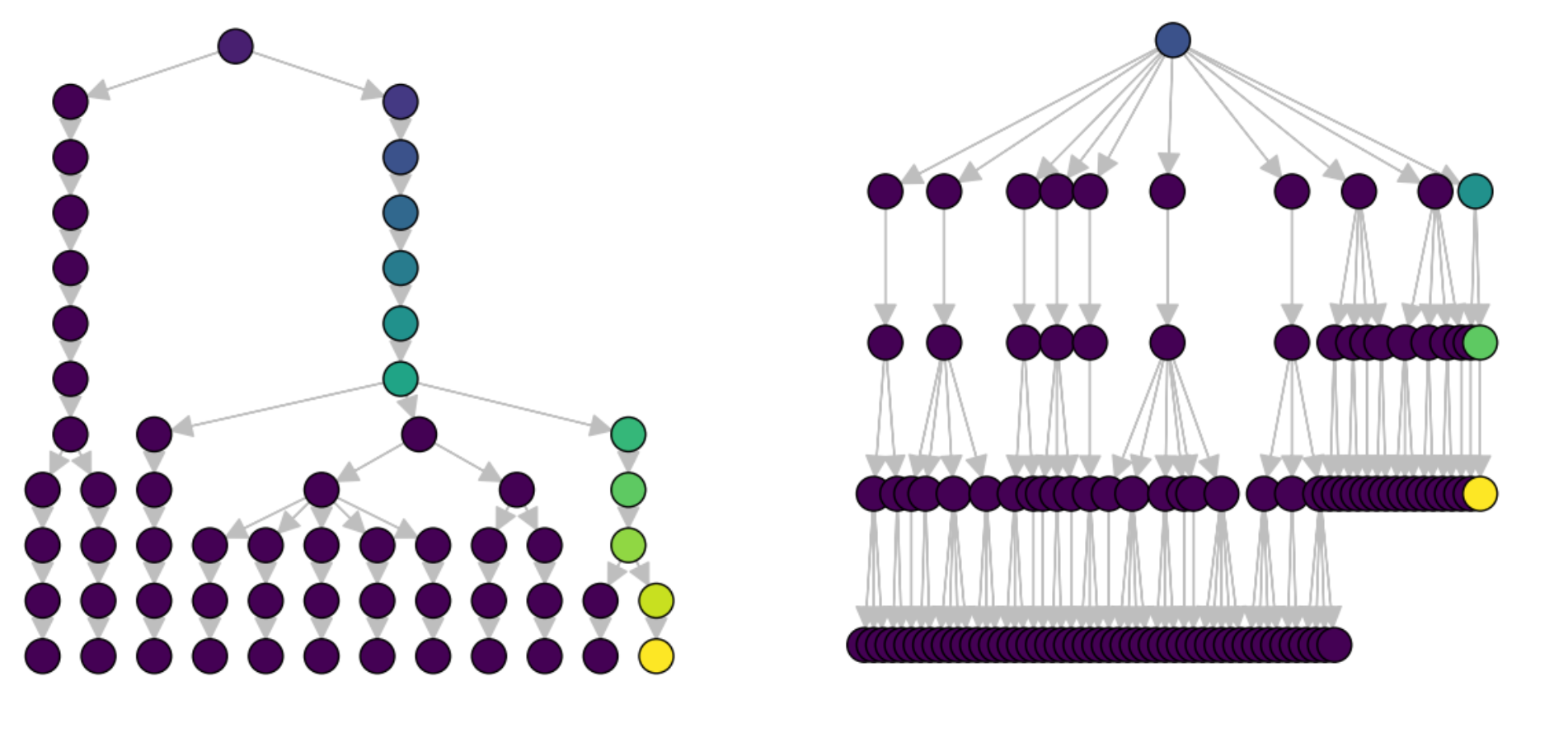}
    \caption{Two trees with labelled paths.}
    \label{fig:bushy}
\end{figure}

\section{Gemini}

The title was suggested by Gemini based on our introduction. 
It was justified like this. ``It is my personal favorite. It twists the classic ``needle in a haystack'' idiom (which everyone in data science knows) to perfectly describe your specific contribution: you aren't looking for a single point, but a linear sequence (a thread).''

\begin{center}
\includegraphics[width=0.8\linewidth]{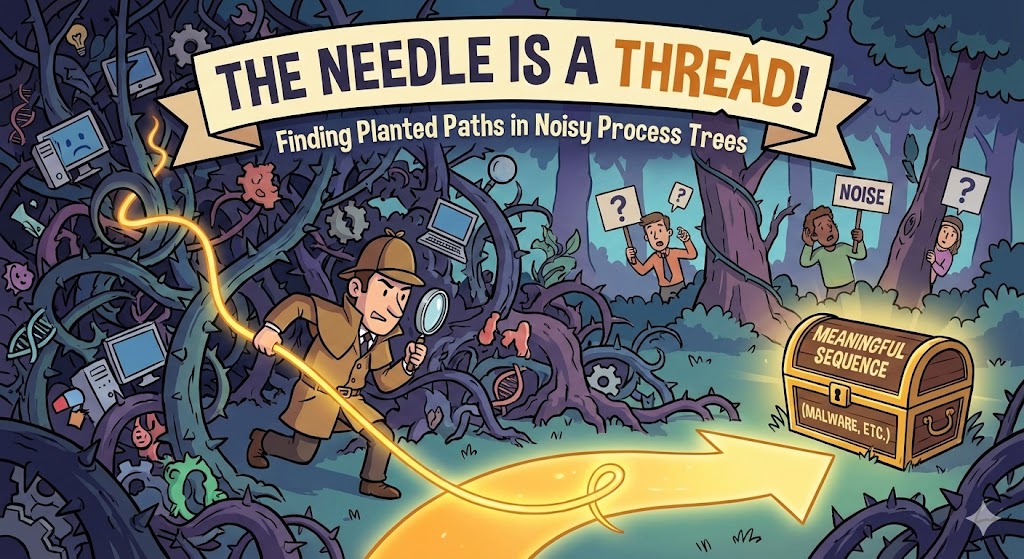}
\end{center}

\end{document}